\documentclass{article}

\usepackage{epsf,hyperref}
\usepackage{amssymb,ComplexSystems}
\usepackage{amsmath,mathtools,graphicx,booktabs,xcolor}

\hypersetup{colorlinks=true,linkcolor=blue!60!black,citecolor=blue!60!black,urlcolor=blue!60!black}

\newcommand{\GF}{\mathbb{F}}
\newcommand{\Z}{\mathbb{Z}}
\newcommand{\R}{\mathbb{R}}
\newcommand{\xor}{\oplus}
\newcommand{\pop}{\operatorname{popcount}}
\newcommand{\Sens}{\operatorname{Sens}}

\newtheorem{remark}{Remark}

\begin{document}
	
	\title{Symmetric Nonlinear Cellular Automata\\
		as Algebraic References for Rule~30}
	
	\author{\authname{E. Chan-L\'opez}\\[2pt]
		\authadd{Divisi\'on Acad\'emica de Ciencias B\'asicas}\\
		\authadd{Universidad Ju\'arez Aut\'onoma de Tabasco}\\
		\authadd{Cunduac\'an, Tabasco, Mexico}
		\and
		\authname{A. Mart\'\i n-Ruiz}\\[2pt]
		\authadd{Instituto de Ciencias Nucleares}\\
		\authadd{Universidad Nacional Aut\'onoma de M\'exico}\\
		\authadd{Mexico City, Mexico}
	}
	
	\markboth{E. Chan-L\'opez and A. Mart\'\i n-Ruiz}
	{Symmetric Nonlinear CA as References for Rule 30}
	
	\maketitle
	
	\begin{abstract}
		A comparative algebraic framework for elementary cellular automata is developed, centered on the role of spatial symmetry. The primary object of study is Rule~22, the elementary cellular automaton with algebraic normal form $g(a,b,c)=a\xor b\xor c\xor abc$ over $\GF_2$, the simplest rule combining full $S_3$ symmetry with genuine nonlinearity. It is proved that the single-seed orbit of Rule~22 admits the closed form $Q_{2j}(x)=(x^2+x^{-2})^j$ and $Q_{2j+1}(x)=(x^{-1}+1+x)(x^2+x^{-2})^j$ for the row generating polynomials over $\GF_2$. Three consequences follow: an exact support-set cardinality $|S_m|=2^{\pop(\lfloor m/2\rfloor)}\cdot 3^{m\bmod 2}$, a two-step recursive construction of the support sets, and, in a formal continuum approximation, a parabolic reaction--diffusion equation $\partial_m u=u_{xx}+2u+u^3$. Rule~22 is then used as a symmetric reference for Rule~30. The symmetry-breaking deviation $\epsilon(m)=|S_m^{(30)}|-|S_m^{(22)}|$ between full-row support cardinalities is empirically consistent with a power law $m^b$ ($b\approx 1.11$), and is non-positive exactly at the Mersenne indices $m=2^k-1$. A mechanism for the apparent randomness of Rule~30's center column is identified through the left-permutive structure and asymmetric Boolean sensitivity profile.
	\end{abstract}
	\vspace{1em}

	\section{Introduction}
	\label{sec:intro}
	
	Elementary cellular automata (ECA) are the simplest class of one-dimensional cellular automata: a bi-infinite row of binary cells updates synchronously according to a local rule $g\colon\{0,1\}^3\to\{0,1\}$ that depends on a cell and its two nearest neighbors. Despite their minimal description length, the 256 possible rules exhibit the full spectrum of dynamical behavior, from trivial fixed points to patterns that support Turing-complete computation \cite{cook2004,wolfram2002}.
	
	Wolfram's classification of ECA into four classes---from simple convergence (Class~1) through nested self-similar patterns (Class~2) to apparent randomness (Class~3) and complex localized structures (Class~4)---remains a guiding framework for the study of discrete dynamical systems \cite{wolfram1984stat,wolfram2002}. Rule~30, a Class~3 automaton, has attracted particular attention because it generates seemingly random output from a deterministic rule applied to a single-seed initial condition. Three fundamental open questions about Rule~30 have been posed by Wolfram \cite{wolfram_prizes}: whether its center column is truly random, whether it is eventually periodic, and whether its individual cell values can be computed faster than by explicit simulation. These questions are intimately connected to the broader concept of computational irreducibility \cite{wolfram2002}.
	
	Rule~30 has also attracted algebraic attention, notably through the decomposition of its local rule into an algebraic normal form (ANF) over $\GF_2$ \cite{wolfram1984alg} and through computational discussions of its support structure on Mathematica Stack Exchange and Wolfram Community \cite{mse2026,chan2026,tigran2026}, which motivated the algebraic perspective adopted here.
	
	In this paper, a different approach is taken: \textbf{symmetric comparison}. Rule~22 is analyzed, an ECA that shares the same linear part $e_1=a\xor b\xor c$ as Rule~30 but differs in its nonlinear correction---the symmetric cubic $abc$ versus the asymmetric quadratic $bc$. This single algebraic difference has profound consequences:
	
	\begin{itemize}
		\item Rule~22 admits an exact closed form for its single-seed orbit, from which a cardinality formula, a recursive construction, and a parabolic continuum approximation follow---all open for Rule~30.
		\item The deviation between the two rules quantifies the effect of symmetry breaking and is empirically consistent with $\epsilon(m)\sim m^{1.11}$.
		\item The mechanism underlying Rule~30's apparent randomness is traced to the left-permutive decomposition $g=a\xor h(b,c)$ and its asymmetric sensitivity profile, which is the discrete counterpart of the first-order transport term in the continuum approximation.
	\end{itemize}
	
	The study of Boolean functions over $\GF_2$ and their influence on dynamical properties has a long history in the theory of cellular automata \cite{wolfram1984alg}, symbolic dynamics \cite{lind_marcus}, and the analysis of pseudorandom sequences \cite{wolfram2002}. The present work connects these perspectives through the lens of spatial symmetry and its algebraic consequences.

	\section{Algebraic Normal Form and Symmetry}
	\label{sec:anf}
	
	Rule~22 has Wolfram code $22=(00010110)_2$. By M\"obius inversion on the Boolean lattice $\{0,1\}^3$ (see, for example, \cite{wolfram1984alg,crama_hammer}), the ANF over $\GF_2$ is computed:
	\begin{equation}
		\label{eq:anf22}
		g_{22}(a,b,c) = a\xor b\xor c\xor abc.
	\end{equation}
	This is the sum of the first and third elementary symmetric polynomials: $g_{22}=e_1\xor e_3$. The output is~1 if and only if exactly one input is~1.
	
	\begin{definition}
		\label{def:permutive}
		A Boolean function $g\colon\{0,1\}^3\to\{0,1\}$ is \textbf{left-permutive} if, for every fixed $(b,c)$, the map $a\mapsto g(a,b,c)$ is a bijection; \textbf{center-} and \textbf{right-permutive} are defined analogously in the variables $b$ and $c$. It has \textbf{full $S_3$ symmetry} if $g(a_{\sigma(1)},a_{\sigma(2)},a_{\sigma(3)})=g(a_1,a_2,a_3)$ for all $\sigma\in S_3$. A rule is \textbf{quiescent} if $g(0,0,0)=0$.
	\end{definition}
	
	Rules~22, 30 and~150 are quiescent; Rule~135 is not. Quiescence is what makes the single-seed support sets of Section~\ref{sec:support} finite, and it is assumed wherever support sets are used.
	
	\begin{proposition}
		\label{prop:symmetry}
		Rule~22 has full $S_3$ symmetry. Rule~30, with ANF $g_{30}=a\xor b\xor c\xor bc$, is left-permutive but lacks $S_3$ symmetry.
	\end{proposition}
	
	\proof Both $e_1$ and $e_3$ are symmetric polynomials, so $g_{22}$ is $S_3$-invariant.
	
	For Rule~30, left-permutivity is immediate from the decomposition $g_{30}=a\xor h(b,c)$ with $h(b,c)=b\xor c\xor bc$, since $a\mapsto a\xor h$ is a bijection for every fixed $(b,c)$. The asymmetry is evident directly from the ANF: the polynomial contains the monomial $bc$ but not $ab$. Since these are exchanged by the transposition $a\leftrightarrow c$, the rule fails to be invariant under the action of $S_3$.
	\endproof
	
	\begin{proposition}
		\label{prop:sym_perm_dichotomy}
		Among the 256 elementary cellular automata, exactly two are simultaneously $S_3$-symmetric and left-permutive: Rule~150, with ANF $g_{150}=a\xor b\xor c$, and its output complement Rule~105, with ANF $g_{105}=1\xor a\xor b\xor c$. In particular, every nonlinear $S_3$-symmetric ECA fails to be permutive in any variable.
	\end{proposition}
	
	\proof An $S_3$-invariant function $g\colon\{0,1\}^3\to\{0,1\}$ is constant on each Hamming-weight orbit, hence determined by the four values $o_w=g(a,b,c)$ for any $(a,b,c)$ with $a+b+c=w$, $w\in\{0,1,2,3\}$; there are therefore exactly $2^4=16$ such rules. Left-permutivity requires $g(0,b,c)\neq g(1,b,c)$ for every $(b,c)\in\{0,1\}^2$. Taking $(b,c)$ of weight $0,1,2$ in turn, this translates into $o_0\neq o_1$, $o_1\neq o_2$, $o_2\neq o_3$. Over $\GF_2$ this forces alternation, leaving only
	\[
	(o_0,o_1,o_2,o_3)\in\big\{(0,1,0,1),\,(1,0,1,0)\big\},
	\]
	corresponding to Rules~150 and~105 respectively. By $S_3$ symmetry the same alternation condition is equivalent to right- and center-permutivity, so the three permutivities coincide on the $S_3$-symmetric class. Both Rules~150 and~105 are affine over $\GF_2$, so any $S_3$-symmetric ECA whose ANF contains a nonlinear monomial is non-permutive in every coordinate. \endproof
	
	\begin{remark}
		Of the 16 $S_3$-symmetric ECA, four are affine, namely $g\in\{0,\,1,\,e_1,\,1\xor e_1\}$ (Rules~0, 255, 150 and~105); the remaining 12 are nonlinear and hence, by Proposition~\ref{prop:sym_perm_dichotomy}, non-permutive. Rule~22 illustrates this: the $S_3$-invariant ANF $g_{22}=e_1\xor e_3$ has orbit values $(o_0,o_1,o_2,o_3)=(0,1,0,0)$, which fails the alternation condition at $w=2\to 3$. Concretely, $g_{22}(0,1,1)=g_{22}(1,1,1)=0$, so $a\mapsto g_{22}(a,1,1)$ is constant. Within the $S_3$-symmetric class, permutivity is available only at the cost of linearity. Symmetry and permutivity are therefore independent structural properties, and the tractable behavior analyzed in this work is driven by $S_3$ symmetry alone.
	\end{remark}
	
	The key consequence for the continuum approximation (Section~\ref{sec:pde}) is that $S_3$ symmetry forces the expansion to be even in the spatial variable, eliminating the first-order transport term and converting the limiting equation from hyperbolic to parabolic type \cite{smoller,murray}.

	\section{Support Sets and the Closed Form}
	\label{sec:support}
	
	Throughout this section $\eta_m$ denotes the Rule~22 configuration at time $m$ evolved from the single-seed initial condition $\eta_0=\delta_0$.
	
	\begin{definition}
		\label{def:support}
		The \textbf{support set} at time $m$ is the full-row support
		\[
		S_m=\{r\in\Z:\eta_m(r)=1\},
		\]
		and the \textbf{right-half support set} is $S_m^{+}=S_m\cap\Z_{\geq 0}$. The associated generating polynomials are
		\[
		Q_m(x)=\sum_{r\in S_m}x^{r}\in\GF_2[x,x^{-1}],
		\qquad
		P_m(x)=\sum_{r\in S_m^{+}}x^{r}\in\GF_2[x].
		\]
	\end{definition}
	
	All cardinality statements below refer to the full-row set $S_m$; the right-half set is used only in Section~\ref{sec:genpoly}, where a genuine polynomial is required. The relation between the two is recorded in Corollary~\ref{cor:half}.
	
	\begin{figure}
		\centerline{\includegraphics[width=0.95\textwidth]{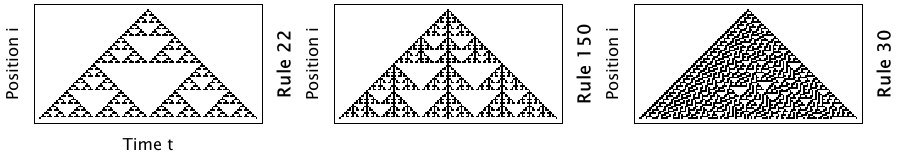}}
		\caption{Spatio-temporal evolution from a single seed for 64 generations. Rule~22 (left) produces a modified Sierpi\'nski triangle with clusters of three consecutive active cells. Rule~150 (center) produces the classical Sierpi\'nski triangle. Rule~30 (right) produces the well-known irregular pattern. The bilateral symmetry of Rules~22 and~150 contrasts with Rule~30's asymmetry.}
		\label{fig:comparison}
	\end{figure}
	
	The whole of the Rule~22 analysis rests on a single structural fact: at even times the active cells are widely separated, so that the cubic term $abc$ cannot fire and the rule acts additively.
	
	\begin{theorem}
		\label{thm:closedform}
		For all $j\geq 0$,
		\begin{equation}
			\label{eq:closedform}
			Q_{2j}(x)=\big(x^{2}+x^{-2}\big)^{j},
			\qquad
			Q_{2j+1}(x)=\big(x^{-1}+1+x\big)\big(x^{2}+x^{-2}\big)^{j},
		\end{equation}
		with all products taken in $\GF_2[x,x^{-1}]$. Moreover $S_{2j}\subseteq 2j+4\Z$, so that any two distinct elements of $S_{2j}$ differ by at least~4.
	\end{theorem}
	
	\proof Induction on $j$. For $j=0$, $Q_0=1$ and $S_0=\{0\}\subseteq 4\Z$.
	
	Assume $Q_{2j}=(x^2+x^{-2})^j$. Expanding, its exponents are $2(2i-j)$ for those $i\in\{0,\dots,j\}$ with $\binom{j}{i}$ odd; each such exponent is congruent to $-2j\equiv 2j\pmod 4$. Hence $S_{2j}\subseteq 2j+4\Z$ and distinct elements differ by at least~4.
	
	\emph{Odd step.} By \eqref{eq:anf22}, $g_{22}(a,b,c)=1$ if and only if $a+b+c=1$. A window $W_i=\{i-1,i,i+1\}$ has width~3, while distinct elements of $S_{2j}$ are at distance at least~4; therefore $|W_i\cap S_{2j}|\leq 1$ for every $i$. Consequently $\eta_{2j+1}(i)=1$ if and only if $W_i$ meets $S_{2j}$, that is, if and only if $i\in\{c-1,c,c+1\}$ for some $c\in S_{2j}$. The triples are pairwise disjoint, so no cancellation occurs over $\GF_2$ and
	\[
	Q_{2j+1}=(x^{-1}+1+x)\,Q_{2j},\qquad |S_{2j+1}|=3\,|S_{2j}|.
	\]
	
	\emph{Even step.} Write $S_{2j+1}=\bigsqcup_{c\in S_{2j}}B_c$ with $B_c=\{c-1,c,c+1\}$. Fix $i$ and count $|W_i\cap S_{2j+1}|$.
	
	If $i\in B_c$ then, because the nearest other block lies at distance at least~4, $W_i$ meets $B_c$ only, and $|W_i\cap B_c|=2$ for $i=c\pm 1$ and $=3$ for $i=c$; in either case the count differs from~1 and $\eta_{2j+2}(i)=0$.
	
	If $i=c+2$ for some $c\in S_{2j}$, then $W_i=\{c+1,c+2,c+3\}$. Here $c+1\in B_c$; the cell $c+2$ lies in no block; and $c+3\in B_{c+4}$ precisely when $c+4\in S_{2j}$. Hence the count equals $1+[\,c+4\in S_{2j}\,]$. Symmetrically, for $i=c-2$ the count equals $1+[\,c-4\in S_{2j}\,]$. Note that $i=c\pm 2$ can never lie inside a block, since that would require another element of $S_{2j}$ within distance~3 of~$c$.
	
	For all remaining $i$, $W_i\cap S_{2j+1}=\emptyset$.
	
	Therefore $\eta_{2j+2}(i)=1$ if and only if $i$ is of the form $c+2$ or $c-2$ for an odd number of $c\in S_{2j}$. Since $c-2=c'+2$ exactly when $c=c'+4$, this is precisely the symmetric difference over $\GF_2$, that is,
	\[
	Q_{2j+2}=(x^{2}+x^{-2})\,Q_{2j}=(x^{2}+x^{-2})^{j+1},
	\]
	which also advances the congruence class from $2j$ to $2j+2$ modulo~4. \endproof
	
	\begin{corollary}
		\label{cor:cardinality}
		For all $m\geq 0$,
		\begin{equation}
			\label{eq:cardinality}
			|S_m| = 2^{\,\pop(\lfloor m/2\rfloor)}\cdot 3^{\,m\bmod 2},
		\end{equation}
		where $\pop(k)$ denotes the number of 1-bits in the binary representation of~$k$.
	\end{corollary}
	
	\proof By Theorem~\ref{thm:closedform}, $|S_{2j}|$ equals the number of $i$ with $\binom{j}{i}$ odd, which by Kummer's carry criterion---equivalently, Lucas' theorem modulo~2 \cite{lucas1878}---equals $2^{\pop(j)}$. The odd step multiplies the cardinality by~3 without cancellation, so $|S_{2j+1}|=3\cdot 2^{\pop(j)}$. Setting $j=\lfloor m/2\rfloor$ gives \eqref{eq:cardinality}. \endproof
	
	The formula has a multiplicative structure indexed by the binary digits of~$m$: each bit at position $k\geq 1$ contributes a factor of~2 when set, while the least significant bit contributes a factor of~3. Corollary~\ref{cor:cardinality} makes precise the sense in which this is a Lucas correspondence: the factor $2^{\pop(\lfloor m/2\rfloor)}$ is exactly the number of odd binomial coefficients in row $\lfloor m/2\rfloor$ of Pascal's triangle, while the alternating base factor~3 comes from the trinomial multiplier $x^{-1}+1+x$ of the odd step.
	
	\begin{corollary}
		\label{cor:half}
		$S_m=-S_m$ for all $m$, and $0\notin S_m$ for $m\geq 2$. Consequently $|S_m^{+}|=\tfrac12|S_m|$ for $m\geq 2$, while $S_0^{+}=\{0\}$ and $S_1^{+}=\{0,1\}$.
	\end{corollary}
	
	\proof Both factors in \eqref{eq:closedform} are invariant under $x\mapsto x^{-1}$, so $Q_m(x)=Q_m(x^{-1})$ and $S_m=-S_m$. For $m=2j$ with $j\geq 1$, membership $0\in S_{2j}$ would require $2(2i-j)=0$ with $\binom{j}{i}$ odd, i.e.\ $i=j/2$ and $\binom{j}{j/2}$ odd; by Kummer's criterion the latter holds only if the addition $j/2+j/2$ produces no carries, which for $j\geq 1$ is impossible. For $m=2j+1$ with $j\geq 1$, every element of $S_{2j+1}$ has the form $c+\delta$ with $c\in S_{2j}$ even and $\delta\in\{-1,0,1\}$, so $0\in S_{2j+1}$ would force $c\in\{-1,0,1\}\cap S_{2j}=\emptyset$. The cardinality statement follows by pairing $r$ with $-r$. \endproof
	
	\begin{figure}
		\centerline{\includegraphics[width=0.9\textwidth]{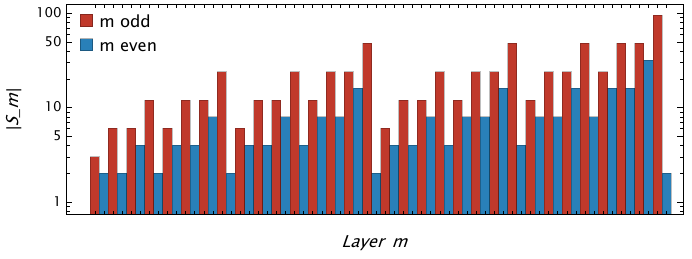}}
		\caption{Full-row support cardinality $|S_m|$ on a logarithmic scale for $m=1,\ldots,64$, as given by Corollary~\ref{cor:cardinality}. Red bars: odd $m$ (factor 3 from the least significant bit). Blue bars: even $m$ (factor 2 only). The pattern resets at each power of~2, where $|S_m|=1$.}
		\label{fig:cardinality}
	\end{figure}

	\section{Recursive Structure of the Support Sets}
	\label{sec:recursion}
	
	The closed form of Theorem~\ref{thm:closedform} can be repackaged as a two-step recursion that cleanly separates the roles of the linear and nonlinear parts of the ANF.
	
	\begin{corollary}
		\label{cor:recursion}
		With $S_0=\{0\}$, the full-row support sets satisfy, for all $m\geq 1$:
		
		(a)~\textbf{Odd step (thickening).} For odd $m$,
		\begin{equation}
			\label{eq:odd}
			S_m = \bigcup_{c\in S_{m-1}}\{c-1,\;c,\;c+1\},
		\end{equation}
		and the union is disjoint.
		
		(b)~\textbf{Even step (decimation).} For $m=2k$ with $k\geq 1$,
		\begin{equation}
			\label{eq:even}
			S_m = 2\cdot\{r\in S_k : r\equiv k\pmod{2}\}.
		\end{equation}
		
		The same two statements hold verbatim with $S_m$ replaced by $S_m^{+}$ for $m\geq 2$.
	\end{corollary}
	
	\proof Part~(a) is the odd step established in the proof of Theorem~\ref{thm:closedform}.
	
	For part~(b), write $m=2k$. If $k=2j$ is even, then every element of $S_k$ is even and the congruence condition is vacuous, so the claim is $S_{4j}=2\cdot S_{2j}$. By Theorem~\ref{thm:closedform} the exponents of $Q_{4j}$ are $2(2i-2j)$ with $\binom{2j}{i}$ odd; by Lucas' theorem \cite{lucas1878} $\binom{2j}{i}$ is odd exactly when every binary digit of $i$ is at most the corresponding digit of $2j$, hence exactly when $i=2i'$ with $\binom{j}{i'}$ odd. The exponents are therefore $4(2i'-j)$, i.e.\ twice the exponents of $Q_{2j}$, as claimed.
	
	If $k=2j+1$ is odd, the congruence selects the odd elements of $S_{2j+1}$, which by part~(a) are exactly $c\pm 1$ for $c\in S_{2j}$. Thus the right-hand side of \eqref{eq:even} is $\{2c\pm 2: c\in S_{2j}\}$. On the other side, using the Frobenius identity in characteristic~2,
		\[
		Q_{4j+2}=(x^{2}+x^{-2})^{2j+1}=(x^{2}+x^{-2})\,(x^{4}+x^{-4})^{j}=(x^{2}+x^{-2})\,Q_{2j}(x^{2}),
		\]
	whose exponents are $2c\pm 2$ for $c\in S_{2j}$. Distinct elements of $S_{2j}$ differ by at least~4, so the values $2c$ differ by at least~8 and the sets $\{2c-2,2c+2\}$ are pairwise disjoint; no $\GF_2$ cancellation occurs, and the two descriptions agree as sets.
	
	The right-half statements follow from Corollary~\ref{cor:half}, since for $m\geq 2$ all the operations involved commute with the reflection $r\mapsto -r$ and $0\notin S_m$. \endproof
	
	The odd step is the algebraic fingerprint of the $abc$ term: when three consecutive cells are all active, the parity function $e_1$ would produce cancellation (output~0), but the cubic correction $abc$ flips the result back to~1, effectively filling in between isolated active cells. The even step implements a self-similar decimation analogous to the scaling properties of the Sierpi\'nski gasket \cite{wolfram2002}. Corollary~\ref{cor:recursion} also yields an $O(\log m)$ algorithm for $|S_m|$ and an $O(|S_m|\log m)$ algorithm for $S_m$ itself, by unwinding \eqref{eq:even} along the binary expansion of~$m$.
	
	\begin{remark}
		\label{rem:additive}
		The mechanism behind Theorem~\ref{thm:closedform} is that on configurations whose active cells are pairwise at distance at least~4, the nonlinear monomial $abc$ never fires and Rule~22 acts as an additive rule. The closed form states that the single-seed orbit remains in this regime at every even time. The reduction of a nonlinear rule to an additive one on sparse configurations is a standard observation; what is specific here is that the Rule~22 orbit never leaves the sparse regime, which is what makes the exact evaluation possible.
	\end{remark}

	\section{Generating Polynomials}
	\label{sec:genpoly}
	
	\begin{proposition}
		\label{prop:degree}
		$\deg P_m = m$ for all $m \geq 1$.
	\end{proposition}
	
	\proof By Theorem~\ref{thm:closedform} the largest exponent of $Q_{2j}$ is $2(2j-j)=2j$, attained at $i=j$ since $\binom{j}{j}=1$ is odd; the odd step adds~1, giving $2j+1$. Hence $\max S_m=m$, and by Corollary~\ref{cor:half} this maximum lies in $S_m^{+}$. \endproof
	
	The same count holds for Rule~30 by the light-cone bound, so the degree of $P_m$ does not by itself separate the two rules; the distinction lies in the arithmetic of the coefficients, which for Rule~22 is governed by Theorem~\ref{thm:closedform}. For Mersenne indices, a clean product form emerges:
	\begin{equation}
		\label{eq:mersenne}
		P_{2^n-1}(x) = x(1+x+x^2)\prod_{j=2}^{n-1}\big(1+x^{2^j}\big),
		\qquad n\geq 2 .
	\end{equation}
	Identity \eqref{eq:mersenne} follows from Theorem~\ref{thm:closedform} together with the classical factorization $(1+y)^{2^{n-1}-1}=\prod_{j=0}^{n-2}(1+y)^{2^{j}}$ in $\GF_2[y]$; it has been verified symbolically for $2\leq n\leq 8$.
	
	\begin{figure}
		\centerline{\includegraphics[width=0.62\textwidth]{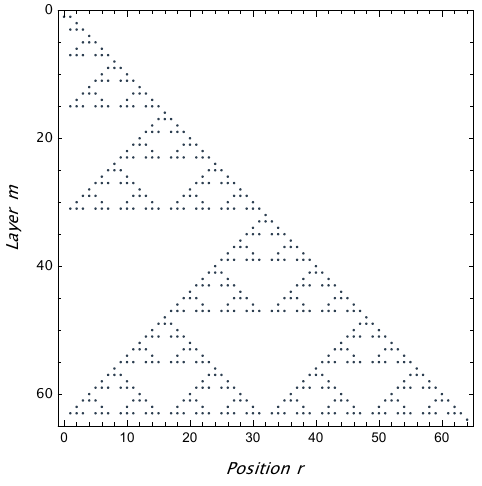}}
		\caption{Right-half support sets $S_m^{+}$ of Rule~22 for $m=1,\ldots,64$. Each mark is a position $r\in S_m^{+}$. The self-similar structure is the geometric content of the decimation step \eqref{eq:even}.}
		\label{fig:support}
	\end{figure}

	\section{Formal Continuum Approximation}
	\label{sec:pde}
	
	The update rule
	\begin{equation}
		\label{eq:update22}
		\eta_{m+1}(i)
		= \eta_m(i-1) \oplus \eta_m(i) \oplus \eta_m(i+1)
		\oplus \eta_m(i-1)\eta_m(i)\eta_m(i+1)
	\end{equation}
	is passed to a continuum equation by the usual formal substitution: the Boolean variables are replaced by a smooth field $u(x,m)$, the exclusive-or by ordinary addition, the products by ordinary products, and the shifts by Taylor expansions in the lattice spacing, which is set to~1. We stress that this is a \emph{formal} continuum approximation and not a rigorous scaling limit: no small parameter is introduced, the range $u\in\{0,1\}$ is relaxed to $\R$, and the truncation is by inspection of derivative order rather than by a controlled estimate. The equation obtained should be read as a structural summary of the ANF, not as an asymptotic theorem.
	
	With $u(x\pm 1)=u\pm u_x+\tfrac12 u_{xx}+O(\partial^3)$ one obtains $u(x{-}1)+u(x)+u(x{+}1)=3u+u_{xx}+O(\partial^4)$ and
	\[
	u(x{-}1)\,u(x)\,u(x{+}1)=u^{3}+u^{2}u_{xx}-u\,u_x^{2}+O(\partial^4).
	\]
	Subtracting $u$ and keeping the leading terms in each structural class yields
	\begin{equation}
		\label{eq:pde}
		\frac{\partial u}{\partial m} = u_{xx}+2u+u^3 .
	\end{equation}
	This is a \textbf{reaction--diffusion equation} of the form $u_m=u_{xx}+f(u)$ with source $f(u)=u(2+u^2)$. Equations of this type arise broadly in mathematical biology and nonlinear wave theory \cite{murray,smoller}.
	
	The salient feature of \eqref{eq:pde} is the absence of the first-order transport term $u_x$, and this absence is exact rather than accidental: the right-hand side of \eqref{eq:update22} is invariant under $i\mapsto -i$, so the formal expansion contains only even-order derivatives at every order. The same computation for Rule~30, whose update is $\eta_{m+1}(i)=\eta_m(i-1)\xor\eta_m(i)\xor\eta_m(i+1)\xor\eta_m(i)\eta_m(i+1)$, gives
	\[
	u(x)u(x{+}1)=u^{2}+u\,u_x+\tfrac12 u\,u_{xx}+O(\partial^3),
	\]
	hence
	\begin{equation}
		\label{eq:pde30}
		\partial_m u - u\,\partial_x u = \big(1+\tfrac12 u\big)\,\partial_{xx}u + 2u+u^{2}.
	\end{equation}
	Writing both rules in the common form $\partial_m u+v(u)\,\partial_x u=\mathcal D(u)\,\partial_{xx}u+\mathcal S(u)$, one has $v\equiv 0$ for Rule~22 and $v(u)=-u\not\equiv 0$ for Rule~30. The structural conclusion---symmetric ANF gives a purely diffusive equation, asymmetric ANF produces a state-dependent transport term---is independent of the particular truncation; the coefficient itself is not, and depends on the particular way the Boolean algebra is relaxed to the reals.
	
	The spatially homogeneous ODE $\dot u=2u+u^3$ associated with \eqref{eq:pde} admits the blow-up solution
	\begin{equation}
		\label{eq:blowup}
		u(m)=\frac{u_0\sqrt{2}\,e^{2m}}{\sqrt{2+u_0^2(1-e^{4m})}}\,,\qquad m^*=\tfrac{1}{4}\ln(1+2/u_0^2).
	\end{equation}
	Setting $u_m=0$ gives the undamped Duffing equation $u''+2u+u^3=0$, integrable via the Jacobi elliptic function $\mathrm{cn}$ \cite{nayfeh_mook}.
	
	\begin{table}
		\centerline{\small\begin{tabular}{|l|c|c|c|c|}
				\hline
				Property & Rule 22 & Rule 30 & Rule 135 & Rule 150 \\
				\hline
				ANF & $a{\xor}b{\xor}c{\xor}abc$ & $a{\xor}b{\xor}c{\xor}bc$ & $1{\xor}a{\xor}bc$ & $a{\xor}b{\xor}c$ \\
				\hline
				Symmetry & $S_3$ & None & None & $S_3$ \\
				\hline
				Quiescent & Yes & Yes & No & Yes \\
				\hline
				$\deg P_m$ & $m$ & $m$ & --- & $m$ \\
				\hline
				Transport $v(u)$ & $0$ & $-u$ & $1-u$ & $0$ \\
				\hline
				$|S_m|$ & $2^{\nu}\cdot 3^{m \bmod 2}$ & Open & --- & $\#\{$odd trinomial coeffs$\}$ \\
				\hline
		\end{tabular}}
		\caption{Structural comparison of ECA rules, with $\nu=\pop(\lfloor m/2\rfloor)$. Symmetry determines whether the formal continuum approximation carries a transport term. Rule~135 is not quiescent ($g(0,0,0)=1$, and the all-zero background is mapped to the all-one fixed point), so the single-seed support sets of Definition~\ref{def:support} are not defined for it and the corresponding entries are left blank; its transport coefficient is computed from the same formal expansion about $u=0$ and is recorded only for comparison. For Rule~150 the row polynomial is $(1+x+x^2)^m$ over $\GF_2$ in the recentered variable.}
		\label{tab:comparison}
	\end{table}

	\section{Symmetry-Breaking Deviation}
	\label{sec:deviation}
	
	Rule~22 is now used as a symmetric reference for Rule~30. Since both rules share the linear part $e_1=a\xor b\xor c$, differing only in the nonlinear correction ($abc$ versus $bc$), the deviation
	\begin{equation}
		\label{eq:deviation}
		\epsilon(m) = |S_m^{(30)}|-|S_m^{(22)}|
	\end{equation}
	isolates the cumulative effect of symmetry breaking. Here both terms are \emph{full-row} cardinalities in the sense of Definition~\ref{def:support}: $|S_m^{(30)}|$ is the number of active cells in row $m$ of Rule~30 from a single seed, and $|S_m^{(22)}|$ is given in closed form by Corollary~\ref{cor:cardinality}. Both rules are quiescent, so both quantities are finite and directly comparable.
	
	\begin{remark}
		\label{rem:mersenne}
		A direct computation for $m\leq 256$ shows that $\epsilon(m)\leq 0$ precisely at the Mersenne indices $m=2^k-1$, with $\epsilon=0$ for $k\leq 3$ and $\epsilon<0$ for $4\leq k\leq 8$. These are exactly the indices at which Corollary~\ref{cor:cardinality} attains its local maxima, $|S_{2^k-1}^{(22)}|=3\cdot 2^{k-1}$, so the sign changes reflect the extremal behavior of the symmetric reference rather than any feature of Rule~30.
	\end{remark}
	
	A log--log regression over the positive values of $\epsilon(m)$ for $m\leq 128$ (121 of the 128 points; the seven exceptions are the Mersenne indices of Remark~\ref{rem:mersenne}) gives
	\begin{equation}
		\label{eq:powerlaw}
		\epsilon(m)\sim m^b,\qquad b\approx 1.11,
	\end{equation}
	empirically consistent with weakly superlinear power-law growth (Figure~\ref{fig:deviation}). This is a numerical fit over a single decade and a half, not a theorem: repeating the regression on the ranges $m\leq 64$ and $m\leq 100$ gives $b\approx 1.11$ and $b\approx 1.16$ respectively, so the exponent should be quoted as $b=1.1\pm 0.05$ over the accessible range. No asymptotic claim is made.
	
	\begin{figure}
		\centerline{\includegraphics[width=0.95\textwidth]{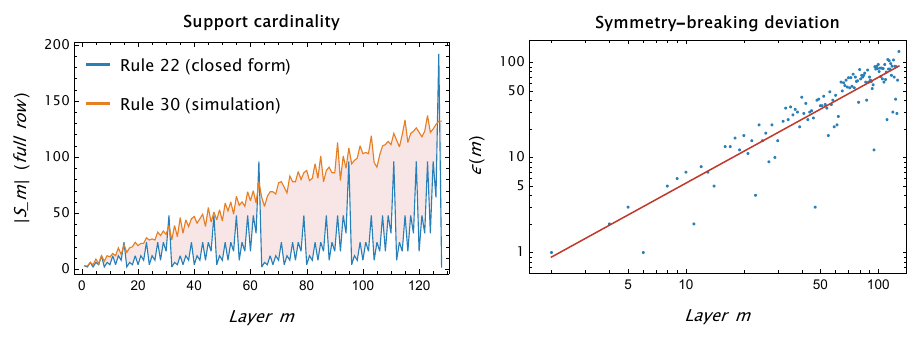}}
		\caption{Left: full-row support cardinalities $|S_m^{(22)}|$ (closed form, Corollary~\ref{cor:cardinality}) and $|S_m^{(30)}|$ (obtained by direct simulation) overlaid. Right: symmetry-breaking deviation $\epsilon(m)$ on log--log axes, restricted to the positive values, with the fitted power law $\epsilon\sim m^{1.11}$.}
		\label{fig:deviation}
	\end{figure}
	
	The continuum picture of Section~\ref{sec:pde} suggests the following heuristic reading. Both rules fit the general form $\partial_m u+v(u)\partial_x u=\mathcal{D}(u)\,\partial_{xx}u+\mathcal{S}(u)$, with $v\equiv 0$ for Rule~22 and $v(u)=-u\not\equiv 0$ for Rule~30. A transport term advects perturbations along characteristics that need not remain parallel, which is the continuum counterpart of the loss of spatial correlation \cite{smoller}. We do not claim that the exponent $b$ is derivable from \eqref{eq:pde30}; the connection is offered as an interpretation of the sign and rough magnitude of $\epsilon$, not as a derivation.

	\section{A Mechanism for Apparent Randomness}
	\label{sec:randomness}
	
	A central open question about Rule~30 is whether its center column is random \cite{wolfram_prizes,wolfram2002}. The left-permutive structure, combined with an asymmetric sensitivity profile, provides a concrete mechanism. Throughout this section we write $c(t)=\eta_t(0)$ for the center-column sequence; in the notation of Definition~\ref{def:support} applied to Rule~30, $c(t)=[\,0\in S_t\,]$.
	
	\subsection{Left-Permutive Decomposition and the XOR Lemma}
	
	Since $g_{30}=a\xor h(b,c)$ with $h(b,c)=b\xor c\xor bc$, the center column satisfies
	\begin{equation}
		\label{eq:xor_decomp}
		\eta_{t+1}(0)=\eta_t(-1)\;\xor\; h\big(\eta_t(0),\,\eta_t(1)\big).
	\end{equation}
	
	\begin{theorem}
		\label{thm:equidistribution}
		Let the initial configuration be i.i.d.\ Bernoulli(1/2). Then $P(\eta_t(0)=1)=1/2$ for all $t\geq 1$ under any left-permutive ECA.
	\end{theorem}
	
	\proof By shift-invariance of the initial product measure, $\eta_t(-1)$ is Bernoulli(1/2). By~\eqref{eq:xor_decomp}, $\eta_{t+1}(0)=\eta_t(-1)\xor Y$ where $Y=h(\eta_t(0),\eta_t(1))$. The XOR lemma states that if $X\sim\mathrm{Bernoulli}(1/2)$ then $X\xor Y\sim\mathrm{Bernoulli}(1/2)$ for any $\{0,1\}$-valued $Y$, even one dependent on~$X$. \endproof
	
	Theorem~\ref{thm:equidistribution} concerns random initial conditions only. It says nothing about the single-seed center column, which is a single deterministic sequence; the passage from one to the other is discussed in Section~\ref{sec:conclusions}. Rule~22 is not permutive, so the argument does not apply to it; the bilateral symmetry of its sensitivity profile follows instead from $S_3$-invariance of $g_{22}$.
	
	\subsection{Asymmetric Sensitivity Profile}
	
	The Boolean sensitivity $\Sens_j(\eta_t(0))$ is the probability, over random initial conditions, that flipping cell $j$ at time~0 changes $\eta_t(0)$. Equivalently, it is the influence of the coordinate $j$ on the Boolean function $\{0,1\}^{\{-t,\dots,t\}}\to\{0,1\}$ sending $\eta_0$ to $\eta_t(0)$, and the totals $\sigma_L+\sigma_R$ below are its total influence, in the sense of Kahn, Kalai and Linial \cite{kkl}. Monte Carlo estimates with 5000 trials at $t=5,\ldots,20$ reveal a pronounced asymmetry (Figure~\ref{fig:sensitivity}).
	
	\begin{figure}
		\centerline{\includegraphics[width=0.95\textwidth]{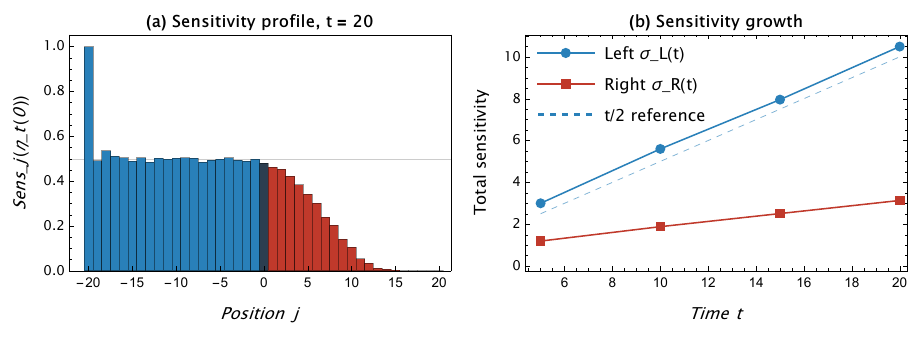}}
		\caption{(a) Estimated sensitivity profile of $\eta_{20}(0)$ for Rule~30 from random initial conditions. Blue bars (left, $j<0$): flat at $\approx 0.5$. Red bars (right, $j>0$): decaying. (b) Growth of total left and right sensitivity over time. Both panels are Monte Carlo estimates, not exact values.}
		\label{fig:sensitivity}
	\end{figure}
	
	One coordinate is exceptional and is determined exactly.
	
	\begin{proposition}
		\label{prop:extremal_influence}
		For any left-permutive ECA and any $t\geq 1$, the leftmost cell of the light cone has influence exactly~1:
		\[
		\Sens_{-t}\big(\eta_t(0)\big)=1 .
		\]
	\end{proposition}
	
	\proof Induction on~$t$. For $t=1$, equation~\eqref{eq:xor_decomp} reads
	\[
	\eta_1(0)=\eta_0(-1)\xor h\big(\eta_0(0),\eta_0(1)\big),
	\]
	whose second summand does not involve $\eta_0(-1)$; flipping that cell therefore flips $\eta_1(0)$. Assume the statement at time~$t$. By shift-invariance it holds at every site, so flipping $\eta_0(-t-1)$ flips $\eta_t(-1)$. In $\eta_{t+1}(0)=\eta_t(-1)\xor h(\eta_t(0),\eta_t(1))$ the arguments $\eta_t(0)$ and $\eta_t(1)$ are determined by $\eta_0$ on $\{-t,\dots,t\}$ and $\{-t+1,\dots,t+1\}$ respectively, neither of which contains $-t-1$. Hence $\eta_{t+1}(0)$ is flipped as well. \endproof
	
	Proposition~\ref{prop:extremal_influence} also gives a second, purely combinatorial proof of Theorem~\ref{thm:equidistribution}: flipping the coordinate $-t$ is an involution of $\{0,1\}^{\{-t,\dots,t\}}$ that exchanges the two level sets of $\eta_t(0)$, so the function is balanced.
	
	Away from that extreme coordinate the left sensitivity is flat at $\Sens_j\approx 1/2$ for $-t<j\leq 0$, while the right sensitivity decays, with measured asymmetry ratio $\sigma_L/\sigma_R\approx 3.3$ at $t=20$. The flat left profile is a consequence of left-permutivity: each left cell enters the computation through the XOR in \eqref{eq:xor_decomp}. The right decay reflects the conditional cancellation $g_{30}(a,1,c)=a\xor 1$, which is independent of~$c$: when the center cell is~1, the right neighbor has no effect.
	
	It is worth recording what this profile is \emph{not}. Since $\eta_t(0)$ is balanced, the theorem of Kahn, Kalai and Linial \cite{kkl} applies and guarantees a coordinate of influence at least $c\log n/n$ with $n=2t+1$ and $c$ an absolute constant. Both Proposition~\ref{prop:extremal_influence} and the measured left profile exceed that bound by a wide margin, so the structure visible in Figure~\ref{fig:sensitivity} is not the generic behavior forced on any balanced Boolean function; it is a specific consequence of permutivity, and it is the asymmetry between the two halves, not the size of the influences, that distinguishes Rule~30 from Rule~22.
	
	\subsection{Connection to the Transport Term}
	
	The asymmetric sensitivity profile is the discrete counterpart of the transport term $v(u)\partial_x u$ in the formal continuum approximation. For Rule~22 (symmetric) both left and right profiles are flat at $\approx 1/2$; for Rule~30 (asymmetric) left perturbations dominate. The measured mutual information
	\[
	I\big(\eta_{20}(-1);(\eta_{20}(0),\eta_{20}(1))\big)\approx 2\times 10^{-5}\;\text{bits}
	\]
	(from $10^5$ trials) is consistent with approximate independence at this sample size, which combined with \eqref{eq:xor_decomp} suggests near-maximal conditional entropy,
	\begin{equation}
		\label{eq:condent}
		H\big(\eta_{t+1}(0)\mid\eta_t(0),\ldots,\eta_1(0)\big)\approx 1\text{ bit}.
	\end{equation}
	Both \eqref{eq:condent} and the mutual-information estimate are empirical; a vanishing sample estimate bounds the true mutual information only up to the estimator's resolution, and no independence claim is made.
	
	Block statistics of the single-seed center column ($N=4096$ steps) give $H_n/n>0.99$ for $n\leq 9$ and full block complexity $p(n)=2^n$ for $n\leq 9$, with $p(n)/2^n$ falling to $0.99$, $0.87$ and $0.63$ at $n=10,11,12$ respectively (Figure~\ref{fig:entropy}). These are finite-sample statistics of one deterministic sequence: at $N=4096$ the counts for $n\geq 12$ are already sparse, so the decay at large $n$ is a sample-size effect and not evidence of structure.
	
	\begin{figure}
		\centerline{\includegraphics[width=0.95\textwidth]{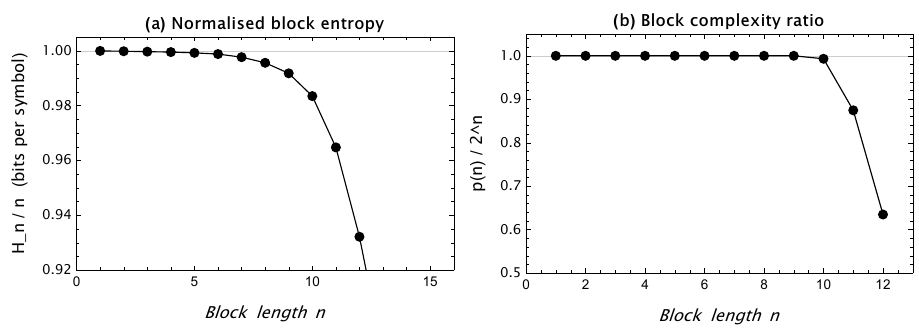}}
		\caption{Statistical properties of Rule~30's center column from a single seed ($N=4096$). (a) Normalized block entropy $H_n/n$ remains near the maximal value~1 for $n\leq 9$. (b) Block complexity ratio $p(n)/2^n$ equals~1 for $n\leq 9$ and decreases thereafter, consistently with the finite sample length.}
		\label{fig:entropy}
	\end{figure}

	\section{Conclusions and Open Problems}
	\label{sec:conclusions}
	
	This paper establishes Rule~22 as a symmetric algebraic reference for Rule~30. The central result is the closed form \eqref{eq:closedform} for the single-seed orbit of Rule~22, proved in Theorem~\ref{thm:closedform}, from which three consequences follow rigorously: the cardinality formula $|S_m|=2^{\pop(\lfloor m/2\rfloor)}\cdot 3^{m\bmod 2}$ (Corollary~\ref{cor:cardinality}), the two-step recursive construction (Corollary~\ref{cor:recursion}), and the linear degree growth $\deg P_m=m$ (Proposition~\ref{prop:degree}). To these are added a formal parabolic continuum approximation $u_m=u_{xx}+2u+u^3$, an empirical symmetry-breaking exponent, and a qualitative randomness mechanism based on the asymmetric sensitivity profile and the XOR decomposition.
	
	These results bear on the three long-standing questions about Rule~30 as follows. On the apparent randomness of the center column, the left-permutive decomposition identifies an algebraic mechanism---the left neighbor enters through an XOR at every step---and Theorem~\ref{thm:equidistribution} makes this precise for random initial conditions; extending it to the single-seed case would require quantifying a mixing rate, which we do not do. On non-periodicity, the center column is $c(t)=[\,0\in S_t\,]$, so eventual periodicity of $c$ is a statement about the sequence of Rule~30 support sets; left-permutivity constrains how periodicity could propagate, but we obtain no non-periodicity result. On computational compression, Corollary~\ref{cor:recursion} shows that $O(\log m)$ evaluation of $|S_m|$ is available when the symmetry is present; whether analogous compression survives for the asymmetric Rule~30 remains tied to the question of computational irreducibility \cite{wolfram2002}.
	
	Several concrete questions remain open. Does $\epsilon(m)$ admit an exact recursion, or an asymptotic description finer than the empirical power law \eqref{eq:powerlaw}? Is the sign pattern of Remark~\ref{rem:mersenne} exact for all $k$? More broadly, the 12 nonlinear $S_3$-symmetric ECA identified in Proposition~\ref{prop:sym_perm_dichotomy} form a natural laboratory for these tools: the argument of Theorem~\ref{thm:closedform} uses only the sparsity invariant $S_{2j}\subseteq 2j+4\Z$ and the vanishing of the nonlinear monomial on separated configurations, and it would be of interest to determine for which of these rules an analogous invariant holds.
	
	These results suggest that symmetry may act as a unifying structural principle governing both algebraic tractability and emergent randomness in cellular automata. In this view, apparent complexity arises not from a lack of underlying rules, but rather from the breaking of symmetries that would otherwise enforce structural regularity.
	
	\section{Supplementary Wolfram Language Material}
	
	The computational material supporting this work is provided as a
	Wolfram Language script,
	\texttt{Rule22Rule30\_Supplement.wl}. The script is designed to
	reproduce the computational results and figures reported in this
	paper, including the Rule~22 support-set calculations, the closed-form
	verification, the recursive constructions, the comparison with
	Rule~30, the symmetry-breaking deviation, and the numerical analyses
	of sensitivity, mutual information, and block statistics.
	
	The script can be executed directly from a Wolfram Language notebook
	using \texttt{Get}. For example, if the notebook and the supplementary
	script are located in the same project directory, the script can be
	invoked by
	
	\begin{verbatim}
		Get["Rule22Rule30_Supplement.wl"]
	\end{verbatim}
	
	or equivalently by
	
	\begin{verbatim}
		<< Rule22Rule30_Supplement.wl
	\end{verbatim}
	
	The script determines its output directory from the location of the
	script itself through
	
	\begin{verbatim}
		outputDirectory = If[SameQ[$InputFileName, ""],
		Directory[],
		DirectoryName @ $InputFileName
		];
	\end{verbatim}
	
	When the script is loaded with \texttt{Get}, the value of
	\texttt{\$InputFileName} identifies the script being evaluated. The
	exported figures and computational outputs are then written to the
	directory containing the script. If the code is evaluated directly
	from a notebook, \texttt{\$InputFileName} is empty and
	\texttt{Directory[]} is used instead.
	
	This design makes the computational material independent of the
	location from which the notebook is launched and avoids hard-coded
	absolute paths. The generated outputs can therefore be reproduced by
	placing the supplementary script in a project directory and invoking
	it with \texttt{Get}. The script does not depend on the previously
	generated manuscript figures; all quantities used to produce the
	figures are recomputed from the definitions and evolution rules given
	in the present work.
	
	\section*{Acknowledgments}
	The authors express their gratitude to \href{https://community.wolfram.com/web/tigrannersissian}{Tigran Nersissian} and the user \href{https://mathematica.stackexchange.com/users/9469/yarchik}{yarchik} for valuable discussions and algebraic insights on the Wolfram Community and Mathematica Stack Exchange platforms, which informed and helped motivate the symmetry--based perspective developed in this work. The authors also thank Matthijs Ruijgrok (Mathematics Department, Utrecht University) for valuable insights concerning the classification of Rule 22.


\begin{thebibliography}{99}
		
		\bibitem{wolfram2002}
		S. Wolfram, \textit{A New Kind of Science}, Champaign, IL: Wolfram Media, 2002.
		
		\bibitem{cook2004}
		M. Cook, ``Universality in Elementary Cellular Automata,'' \textit{Complex Systems}, \textbf{15}(1), 2004 pp.~1--40.
		
		\bibitem{wolfram1984stat}
		S. Wolfram, ``Statistical Mechanics of Cellular Automata,'' \textit{Reviews of Modern Physics}, \textbf{55}(3), 1983 pp.~601--644. doi:10.1103/RevModPhys.55.601.
		
		\bibitem{wolfram1984alg}
		S. Wolfram, ``Algebraic Properties of Cellular Automata,'' \textit{Communications in Mathematical Physics}, \textbf{93}(2), 1984 pp.~219--258. doi:10.1007/BF01223745.
		
		\bibitem{wolfram_prizes}
		S. Wolfram. ``The Wolfram Rule~30 Prizes.'' (Jun 1, 2019)\\
		\url{https://rule30prize.org/}.
		
		\bibitem{mse2026}
		T. Nersissian. ``Rule~30: Finding a Closed Formula for the $S_m$ Subset Recurrence'' from Mathematica Stack Exchange. Question 318912 (2026).\\
		\url{https://mathematica.stackexchange.com/questions/318912/}.
		
		\bibitem{chan2026}
		E. Chan-L\'opez. Answer to ``Rule~30: Finding a Closed Formula for the $S_m$ Subset Recurrence'' from Mathematica Stack Exchange. Question 318912 (2026).\\
		\url{https://mathematica.stackexchange.com/questions/318912/}.
		
		\bibitem{tigran2026}
		T. Nersissian, ``Rule 30 exact binomial--Lucas lifting (part II): generating polynomials, continuous PDE limits, and symmetry classification of elementary cellular automata'', Wolfram Community, 2026.
		\url{https://community.wolfram.com/groups/-/m/t/3671492}
		
		\bibitem{crama_hammer}
		Y. Crama and P. L. Hammer, \textit{Boolean Functions: Theory, Algorithms, and Applications}, Cambridge: Cambridge University Press, 2011. doi:10.1017/CBO9780511852008.
		
		\bibitem{lind_marcus}
		D. Lind and B. Marcus, \textit{An Introduction to Symbolic Dynamics and Coding}, 2nd~ed., Cambridge: Cambridge University Press, 2021. doi:10.1017/9781108899727.
		
		\bibitem{lucas1878}
		\'E. Lucas, ``Th\'eorie des Fonctions Num\'eriques Simplement P\'eriodiques,'' \textit{American Journal of Mathematics}, \textbf{1}(2), 1878 pp.~184--196. doi:10.2307/2369308.
		
		\bibitem{smoller}
		J. Smoller, \textit{Shock Waves and Reaction--Diffusion Equations}, 2nd~ed., New York: Springer-Verlag, 1994. doi:10.1007/978-1-4612-0873-0.
		
		\bibitem{murray}
		J. D. Murray, \textit{Mathematical Biology I: An Introduction}, 3rd~ed., New York: Springer, 2002. doi:10.1007/b98868.
		
		\bibitem{nayfeh_mook}
		A. H. Nayfeh and D. T. Mook, \textit{Nonlinear Oscillations}, New York: Wiley, 1979.
		
		\bibitem{kkl}
		J. Kahn, G. Kalai, and N. Linial, ``The Influence of Variables on Boolean Functions,'' in \textit{Proceedings of the 29th Annual Symposium on Foundations of Computer Science (FOCS~1988)}, White Plains, NY, Washington, DC: IEEE, 1988 pp.~68--80. doi:10.1109/SFCS.1988.21923.
		
	\end{thebibliography}
\end{document}